\documentclass[publicdomain]{eptcs}
\providecommand{\event}{MCU 2013}
\usepackage{makeidx}
\usepackage{amssymb}
\usepackage{amsfonts}
\usepackage{amsmath}
\usepackage{latexsym}
\input{tcilatex}

\def\titlerunning{One-dimensional Array Grammars and 
P Systems with Array Insertion and Deletion Rules} 
\def\authorrunning{Rudolf Freund, Sergiu Ivanov, Marion Oswald, and 
K.G. Subramanian}

\begin{document}

\title{One-dimensional Array Grammars and \\
P Systems with Array Insertion and Deletion Rules}
\author{Rudolf Freund 
\institute{
Technische Universit\"{a}t Wien, Institut f\"{u}r Computersprachen\\
Favoritenstr. 9, A-1040 Wien, Austria}
\email{rudi@emcc.at}
\and Sergiu Ivanov 
\institute{
Laboratoire d'Algorithmique, Complexit\'{e} et Logique,
Universit\'{e} Paris Est \\
Cr\'{e}teil Val de Marne, 61, Av. G\'{e}n. de Gaulle, 
94010 Cr\'{e}teil, France}
\email{sergiu.ivanov@u-pec.fr}
\and Marion Oswald
\institute{
Technische Universit\"{a}t Wien, Institut f\"{u}r Computersprachen\\
Favoritenstr. 9, A-1040 Wien, Austria}
\email{marion@emcc.at}
\and  K.G. Subramanian
\institute{
School of Computer Sciences, Universiti Sains Malaysia,
11800 Penang, Malaysia} 
\email{kgsmani1948@yahoo.com}
} 

\index{Freund, Rudolf} 
\index{Ivanov, Sergiu} 
\index{Oswald, Marion} 
\index{Subramanian, K.~G.}

\maketitle

\begin{abstract}
We consider the (one-dimensional) array counterpart of contextual as well as
insertion and deletion string grammars and consider the operations of array
insertion and deletion in array grammars. First we show that the emptiness
problem for P systems with (one-dimensional) insertion rules is undecidable.
Then we show computational completeness of P systems using (one-dimensional)
array insertion and deletion rules even of norm one only. The main result of
the paper exhibits computational completeness of one-\-di\-men\-sio\-nal
array grammars using array insertion and deletion rules of norm at most two.
\end{abstract}

\section{Introduction}

In the string case, the insertion operation was first considered in~\cite%
{Galiuk81,Haussler82,Haussler83} and after that related insertion and
deletion operations were investigated, e.g., in~\cite{Kari,KPTY97}. Based on
linguistic motivations, checking of insertion contexts was considered in~%
\cite{Marcus69} with introducing \textit{contextual} grammars; these
contextual grammars start from a set of strings (\textit{axioms}), and new
strings are obtained by using rules of the form $(s,c)$, where $s$ and $c$
are strings to be interpreted as inserting $c$ in the context of $s$, either
only at the ends of strings (\textit{external }case,\textit{\ }\cite%
{Marcus69}) or in the \textit{interior} of strings (\cite{Paun80}). The
fundamental difference between contextual grammars and Chomsky grammars is
that in contextual grammars we do not \textit{rewrite} symbols, but we only 
\textit{adjoin} symbols to the current string, i.e., contextual grammars are
pure grammars. Hence, among the variants of these grammars as, for example,
considered in \cite%
{Ehrenfeucht95,Ehrenfeucht96,Ehrenfeucht98,Paun94,Paun95,Paun97}, the
variant where we can retain only the set of strings produced by blocked
derivations, i.\thinspace e., derivations which cannot be continued, is of
special importance. This corresponds to the maximal mode of derivation
(called t-mode) in cooperating grammar systems (see \cite{Csuhaj94}) as well
as to the way results in P systems are obtained by halting computations; we
refer the reader to~\cite{pbook,phandbook} and to the web page~\cite{ppage}
for more details on P systems.

With the length of the contexts and/or of the inserted and deleted strings
being big enough, the insertion-deletion closure of a finite language leads
to computational completeness. There are numerous results establishing the
descriptional complexity parameters sufficient to achieve this goal; for an
overview of this area we refer to~\cite{VerlanH,Verlan2010}. In \cite%
{Freundetal2012} it was shown that computational completeness can also be
obtained with using only insertions and deletions of just one symbol at the
ends of a string using the regulating framework of P systems, where the
application of rules depends on the membrane region.

The contextual style of generating strings was extended to $d$%
-\-di\-men\-sio\-nal arrays in a natural way (see 
\cite{Freundetal2007,Kamala2000}): a
contextual array rule is a pair $\left( s,c\right) $ of two arrays to be
interpreted as inserting the new subarray $c$ in the context of the array $s$
provided that the positions where to put $c$ are not yet occupied by a
non-blank symbol. With retaining only the arrays produced in maximal
derivations, interesting languages of two-dimensional arrays can be
generated. In \cite{Fernauetal2005}, it was shown that every recursively
enumerable one-dimensional array language can be characterized as the
projection of an array language generated by a two-dimensional contextual
array grammar using rules of norm one only (the norm of a contextual array
rule $\left( s,c\right) $ is the maximal distance between two positions in the 
union of the two finite arrays $s$ and $c$). A contextual array rule $\left(
s,c\right) $ can be interpreted as array insertion rule; by inverting the
meaning of this operation, we get an array deletion rule $\left( s,c\right) $
deleting the subarray $c$ in the relative context of the subarray $s$. In 
\cite{Fernauetal2013}, contextual array rules in P systems are considered. P
systems using array insertion and deletion rules were investigated in \cite%
{Freundetal2013}, especially for the two-dimensional case, proving
computational completeness with using array insertion and deletion rules
even of norm one only.

In this paper, we focus on the one-dimensional case. First we show that the
emptiness problem for P systems using one-dimensional contextual array rules
is undecidable. We adapt the proof from \cite{Freundetal2013} for proving
the computational completeness of P systems using array insertion and
deletion rules even of norm one only. The main result of the paper exhibits
computational completeness of one-\-di\-men\-sio\-nal array grammars using
array insertion and deletion rules of norm at most two.

\section{Definitions and Examples}

The set of integers is denoted by $\mathbb{Z}$, the set of non-negative
integers by $\mathbb{N}$. An \textit{alphabet }$V$ is a finite non-empty set
of abstract \textit{symbols}. Given $V$, the free monoid generated by $V$
under the operation of concatenation is denoted by $V^{\ast }$; the elements
of $V^{\ast }$ are called strings, and the \textit{empty string} is denoted
by $\lambda $; $V^{\ast }\setminus \left\{ \lambda \right\} $ is denoted by $%
V^{+}$. Each string $w\in T^{+}$ can be written as $w\left( 1\right) \ldots
w\left( \left\vert w\right\vert \right) $, where $\left\vert w\right\vert $
denotes the length of $w$. The family of recursively enumerable string
languages is denoted by $RE$. For more details of formal language theory the
reader is referred to the monographs and handbooks in this area such as \cite%
{DassowPaun1989} and \cite{handbook}.

\subsection{A General Model for Sequential Grammars}

In order to be able to introduce the concept of membrane systems (P systems)
for various types of objects, we first define a general model (\cite%
{Freundetal2011}) of a grammar generating a set of terminal objects by
derivations where in each derivation step exactly one rule is applied
(sequential derivation mode) to exactly one object.

\smallskip

A \textit{(sequential) grammar} $G$ is a construct $\left(
O,O_{T},w,P,\Longrightarrow _{G}\right) $ where $O$ is a set of \textit{%
objects}, $O_{T}\subseteq O$ is a set of \textit{terminal objects}, $w\in O$
is the \textit{axiom (start object)}, $P$ is a finite set of \textit{rules},
and $\Longrightarrow _{G}\subseteq O\times O$ is the \textit{derivation
relation} of~$G$. We assume that each of the rules $p\in P$ induces a
relation $\Longrightarrow _{p}\subseteq O\times O$ with respect to $%
\Longrightarrow _{G}$ fulfilling at least the following conditions: (i) for
each object $x\in O$, $\left( x,y\right) \in \ \Longrightarrow _{p}$ for
only finitely many objects $y\in O$; (ii) there exists a finitely described
mechanism (as, for example, a Turing machine) which, given an object $x\in O$%
, computes all objects $y\in O$ such that $\left( x,y\right) \in \
\Longrightarrow _{p}$. A rule $p\in P$ is called \textit{applicable} to an
object $x\in O$ if and only if there exists at least one object $y\in O$
such that $\left( x,y\right) \in \ \Longrightarrow _{p}$; we also write $%
x\Longrightarrow _{p}y$. The derivation relation $\Longrightarrow _{G}$ is
the union of all $\Longrightarrow _{p}$, i.e., $\Longrightarrow _{G}=$ $\cup
_{p\in P}\Longrightarrow _{p}$. The reflexive and transitive closure of $%
\Longrightarrow _{G}$ is denoted by $\overset{\ast }{\Longrightarrow }_{G}$.

In the following we shall consider different types of grammars depending on
the components of $G$, especially on the rules in $P$; these may define a
special type $X$ of grammars which then will be called \textit{grammars of
type }$X$\textit{.}

Usually, the \textit{language generated by }$G$ (in the $\ast $-mode) is the
set of all terminal objects (we also assume $v\in O_{T}$ to be decidable for
every $v\in O$) derivable from the axiom, i.e., $L_{\ast }\left( G\right)
=\left\{ v\in O_{T}\mid w\overset{\ast }{\Longrightarrow }_{G}v\right\} $.
The \textit{language generated by }$G$ \textit{in the t-mode} is the set of
all terminal objects derivable from the axiom in a halting computation,
i.e., $L_{t}\left( G\right) =\left\{ v\in O_{T}\mid \left( w\overset{\ast }{%
\Longrightarrow }_{G}v\right) \wedge \nexists z\left( v\Longrightarrow
_{G}z\right) \right\} $. The family of languages generated by grammars of
type $X$ in the derivation mode $\delta $, $\delta \in \left\{ \ast
,t\right\} $, is denoted by $\mathcal{L}_{\mathcal{\delta }}\left( X\right) $%
. If for every $G$ of type $X$, $G=\left( O,O_{T},w,P,\Longrightarrow
_{G}\right) $, we have $O_{T}=O$, then $X$ is called a \textit{pure} type,
otherwise it is called \textit{extended}.

\subsection{String grammars}

In the general notion as defined above, a \textit{string grammar }$G_{S}$ is
represented as $\left( \left( N\cup T\right) ^{\ast },T^{\ast
},w,P,\Longrightarrow _{P}\right) $ where $N$ is the alphabet of \textit{%
non-terminal symbols}, $T$ is the alphabet of \textit{terminal symbols}, $%
N\cap T=\emptyset $, $w\in \left( N\cup T\right) ^{+}$ is the \textit{axiom}%
, $P$ is a finite set of \textit{string rewriting rules}, and the derivation
relation $\Longrightarrow _{G_{S}}$ is the classic one for string grammars
defined over $V^{\ast }\times V^{\ast }$, with $V=N\cup T$. As classic
types of string grammars we consider string grammars with arbitrary rules of
the form $u\rightarrow v$ with $u\in V^{+}$ and $v\in V^{\ast }$ as well as
context-free rules of the form $A\rightarrow v$ with $A\in N$ and $v\in
V^{\ast }$. The corresponding types of grammars are denoted by $ARB$ and $CF$%
, thus yielding the families of languages $\mathcal{L}\left( ARB\right) $
and $\mathcal{L}\left( CF\right) $, i.e., the family of recursively
enumerable languages $RE$ and the family of context-free languages,
respectively.

In \cite{Freundetal2012}, left and right insertions and deletions of strings
were considered; the corresponding types of grammars using rules inserting
strings of length at most $k$ and deleting strings of length at most $m$ are
denoted by $D^{m}I^{k}$.

\subsection{Array grammars\label{arraygrammars}}

We now introduce the basic notions for $d$-\-di\-men\-sio\-nal arrays and
array grammars in a similar way as in \cite{Freund93,Freundetal2007}. Let~$%
d\in \mathbb{N}$; then a $d$\textit{-\-di\-men\-sio\-nal array} $\mathcal{A}$
over an alphabet~$V$ is a function $\mathcal{A}:\mathbb{Z}^{d}\rightarrow
V\cup \left\{ \#\right\} $, where $shape\left( \mathcal{A}\right) =\left\{
v\in \mathbb{Z}^{d}\mid \mathcal{A}\left( v\right) \neq \#\right\} $ is
finite and $\#\notin V$ is called the \textit{background\/} or \textit{blank
symbol}. We usually write $\mathcal{A}=\left\{ \left( v,\mathcal{A}\left(
v\right) \right) \mid v\in shape\left( \mathcal{A}\right) \right\} $. The
set of all $d$-\-di\-men\-sio\-nal arrays over~$V$ is denoted by $V^{\ast d}$%
. The \textit{empty array\/} in~$V^{\ast d}$ with empty shape is denoted by~$%
\Lambda _{d}$. Moreover, we define $V^{+d}=V^{\ast d}\setminus \left\{
\Lambda _{d}\right\} $. Let $v\in \mathbb{Z}^{d}$, $v=\left( v_{1},\dots
,v_{d}\right) $; the norm of $v$ is defined as $\left\Vert v\right\Vert
=\max \left\{ |v_{i}|\mid 1\leq i\leq d\right\} $. The \textit{translation} $%
\tau _{v}:\mathbb{Z}^{d}\rightarrow \mathbb{Z}^{d}$ is defined by $\tau
_{v}\left( w\right) =w+v$ for all $w\in \mathbb{Z}^{d}$. For any array $%
\mathcal{A}\in V^{\ast d}$ we define $\tau _{v}\left( \mathcal{A}\right) $,
the corresponding $d$-\-di\-men\-sio\-nal array translated by $v$, by $%
\left( \tau _{v}\left( \mathcal{A}\right) \right) \left( w\right) =\mathcal{A%
}\left( w-v\right) $ for all $w\in \mathbb{Z}^{d}$. For a (non-empty) finite
set $W\subset Z^{d}$ the norm of~$W$ is defined as $\left\Vert W\right\Vert
=\max \left\{ \,\left\Vert v-w\right\Vert \mid v,w\in W\,\right\} $. The
vector $\left( 0,\dots ,0\right) \in \mathbb{Z}^{d}$ is denoted by $\Omega
_{d}$.

Usually\thinspace (e.g., see \cite{Cook78,Rosenfeld79,Wang80}) arrays are
regarded as equivalence classes of arrays with respect to linear
translations, i.\thinspace e., only the relative positions of the symbols
different from $\#$ in the plane are taken into account: the equivalence
class $\left[ \mathcal{A}\right] $ of an~array $\mathcal{A}\in V^{\ast d}$
is defined by 
\begin{equation*}
\left[ \mathcal{A}\right] =\left\{ \mathcal{B}\in V^{\ast d}\mid \mathcal{B}%
=\tau _{v}\left( \mathcal{A}\right) \text{ for some }v\in \mathbb{Z}%
^{d}\right\} .
\end{equation*}

\noindent The set of all equivalence classes of $d$-\-di\-men\-sio\-nal
arrays over $V$ with respect to linear translations is denoted by $\left[
V^{\ast d}\right] $ etc.

Let $d_{1},d_{2}\in \mathbb{N}$ with ${d_{1}}<d_{2}$. The natural embedding $%
i_{{d_{1}},d_{2}}:\mathbb{Z}^{{d_{1}}}\rightarrow \mathbb{Z}^{d_{2}}$ is
defined by $i_{{d_{1}},d_{2}}\left( v\right) =\left( v,\Omega _{d_{2}-{d_{1}}%
}\right) $ for all $v\in \mathbb{Z}^{{d_{1}}}$. To a ${d_{1}}$%
-\-di\-men\-sio\-nal array $\mathcal{A}\in V^{+{d_{1}}}$ with $\mathcal{A}%
=\left\{ \,\left( v,\mathcal{A}\left( v\right) \right) \mid v\in shape\left( 
\mathcal{A}\right) \,\right\} $, we assign the $d_{2}$-\-di\-men\-sio\-nal
array $i_{{d_{1}},d_{2}}\left( \mathcal{A}\right) =\left\{ \,\left( i_{{d_{1}%
},d_{2}}\left( v\right) ,\mathcal{A}\left( v\right) \right) \mid v\in
shape\left( \mathcal{A}\right) \,\right\} $; moreover, we have $i_{{d_{1}}%
,d_{2}}\left( \Lambda _{d_{1}}\right) =\Lambda _{d_{2}}$.

Any one-dimensional array $\mathcal{A}=\left\{ \,\left( v,\mathcal{A}\left(
v\right) \right) \mid v\in shape\left( \mathcal{A}\right) \,\right\} $ with $%
shape\left( \mathcal{A}\right) =\left\{ \left( m_{i}\right) \mid 1\leq i\leq
n\right\} $, $m_{1}<\dots <m_{n}$, can also be represented as the sequence $%
\left( m_{1}\right) \mathcal{A}\left( \left( m_{1}\right) \right) \dots
\left( m_{n}\right) \mathcal{A}\left( \left( m_{n}\right) \right) $ and
simply as a string $\mathcal{A}\left( \left( m_{1}\right) \right)
\#^{m_{2}-m_{1}-1}\dots \mathcal{A}\left( \left( m_{n}\right) \right) $; for
example, we may write $\left[ \left( -2\right) a\left( -1\right) a\left(
3\right) a\left( 5\right) a\right] $ and $a^{2}\#^{3}a\#a$ for the array $%
\left[ \left\{ \left( \left( -2\right) ,a\right) ,\left( \left( -1\right)
,a\right) ,\left( \left( 3\right) ,a\right) ,\left( \left( 5\right)
,a\right) \right\} \right] $.

A $\mathit{d}$\textit{-\-di\-men\-sio\-nal array grammar} $G_{A}$ is
represented as 
\begin{equation*}
\left( \left[ \left( N\cup T\right) ^{\ast d}\right] ,\left[ T^{\ast d}%
\right] ,\left[ \mathcal{A}_{0}\right] ,P,\Longrightarrow _{G_{A}}\right)
\end{equation*}%
where $N$ is the alphabet of \textit{non-terminal symbols}, $T$ is the
alphabet of \textit{terminal symbols}, $N\cap T=\emptyset $, $\mathcal{A}%
_{0}\in \left( N\cup T\right) ^{\ast d}$ is the \textit{start array}, $P$ is
a finite set of $d$\textit{-\-di\-men\-sio\-nal array rules} over $V$, $%
V=N\cup T$, and $\Longrightarrow _{G_{A}}\subseteq \left[ \left( N\cup
T\right) ^{\ast d}\right] \times \left[ \left( N\cup T\right) ^{\ast d}%
\right] $ is the derivation relation induced by the array rules in $P$.

A \textquotedblleft classical\textquotedblright\ $\mathit{d}$\textit{%
-\-di\-men\-sio\-nal array rule }$p$ over $V$ is a triple $\left( W,\mathcal{%
A}_{1},\mathcal{A}_{2}\right) $ where $W\subseteq \mathbb{Z}^{d}$ is a
finite set and $\mathcal{A}_{1}$ and $\mathcal{A}_{2}$ are mappings from $W$
to $V\cup \left\{ \#\right\} $. In the following, we will also write $%
\mathcal{A}_{1}\rightarrow \mathcal{A}_{2}$, as $W$ is implicitly given by
the finite arrays $\mathcal{A}_{1},\mathcal{A}_{2}$. The norm of the $d$%
-\-di\-men\-sio\-nal array production $\left( W,\mathcal{A}_{1},\mathcal{A}%
_{2}\right) $ is defined by $\left\Vert \left( W,\mathcal{A}_{1},\mathcal{A}%
_{2}\right) \right\Vert =\left\Vert W\right\Vert $. We say that the array $%
\mathcal{C}_{2}\in V^{\ast d}$ is \textit{directly derivable} from the array 
$\mathcal{C}_{1}\in V^{+d}$ by $\left( W,\mathcal{A}_{1},\mathcal{A}%
_{2}\right) $ if and only if there exists a vector $v\in \mathbb{Z}^{d}$
such that $\mathcal{C}_{1}\left( w\right) =\mathcal{C}_{2}\left( w\right) $
for all $w\in \mathbb{Z}^{d}-\tau _{v}\left( W\right) $ as well as $\mathcal{%
C}_{1}\left( w\right) =\mathcal{A}_{1}\left( \tau _{-v}\left( w\right)
\right) $ and $\mathcal{C}_{2}\left( w\right) =\mathcal{A}_{2}\left( \tau
_{-v}\left( w\right) \right) $ for all $w\in \tau _{v}\left( W\right) ,$
i.\thinspace e., the subarray of $\mathcal{C}_{1}$ corresponding to $%
\mathcal{A}_{1}$ is replaced by $\mathcal{A}_{2}$, thus yielding $\mathcal{C}%
_{2}$; we also write $\mathcal{C}_{1}\Longrightarrow _{p}\mathcal{C}_{2}.$
Moreover we say that the array $\mathcal{B}_{2}\in \left[ V^{\ast d}\right] $
is \textit{directly derivable} from the array $\mathcal{B}_{1}\in \left[
V^{+d}\right] $ by the $d$-\-di\-men\-sio\-nal array production $\left( W,%
\mathcal{A}_{1},\mathcal{A}_{2}\right) $ if and only if there exist $%
\mathcal{C}_{1}\in \mathcal{B}_{1}$ and $\mathcal{C}_{2}\in \mathcal{B}_{2}$
such that $\mathcal{C}_{1}\Longrightarrow _{p}\mathcal{C}_{2}$; we also
write $\mathcal{B}_{1}\Longrightarrow _{p}\mathcal{B}_{2}$.

A $d$-\-di\-men\-sio\-nal array rule $p=\left( W,\mathcal{A}_{1},\mathcal{A}%
_{2}\right) $ in $P$ is called \textit{mo\-no\-to\-nic} if $shape\left( 
\mathcal{A}_{1}\right) \subseteq shape\left( \mathcal{A}_{2}\right) $ and $%
\# $-\textit{context-free} if $shape\left( \mathcal{A}_{1}\right) =\left\{
\Omega _{d}\right\} $; if it is $\#$-\textit{context-free} and, moreover, $%
shape\left( \mathcal{A}_{2}\right) =W$, then $p$ is called \textit{%
context-free.} A $d$-\-di\-men\-sio\-nal array grammar is said to be of type 
$X$, $X\in \left\{ d\text{-}ARBA,d\text{-}MONA,d\text{-}\#\text{-}CFA,d\text{%
-}CFA\right\} $ if every array rule in $P$ is of the corresponding type, the
corresponding families of array languages of equivalence classes of $d$%
-\-di\-men\-sio\-nal arrays by $d$-\-di\-men\-sio\-nal array grammars are
denoted by $\mathcal{L}_{\mathcal{\ast }}\left( X\right) $. These families
form a Chomsky-like hierarchy, i.e., $\mathcal{L}_{\mathcal{\ast }}\left( d%
\text{-}CFA\right) \subsetneqq \mathcal{L}_{\mathcal{\ast }}\left( d\text{-}%
MONA\right) \subsetneqq \mathcal{L}_{\mathcal{\ast }}\left( d\text{-}%
ARBA\right) $ and $\mathcal{L}_{\mathcal{\ast }}\left( d\text{-}CFA\right)
\subsetneqq \mathcal{L}_{\mathcal{\ast }}\left( d\text{-}\#\text{-}%
CFA\right) \subsetneqq \mathcal{L}_{\mathcal{\ast }}\left( d\text{-}%
ARBA\right) $.

Two $d$-\-di\-men\-sio\-nal arrays $\mathcal{A}$ and $\mathcal{B}$ in $\left[
V^{\ast d}\right] $ are called \textit{shape-equivalent} if and only if $%
shape\left( \mathcal{A}\right) =shape\left( \mathcal{B}\right) $. Two $d$%
-\-di\-men\-sio\-nal array languages $L_{1}$ and $L_{2}$ from $\left[
V^{\ast d}\right] $ are called shape-equivalent if and only if $\left\{
shape\left( \mathcal{A}\right) \mid \mathcal{A}\in L_{1}\right\} =\left\{
shape\left( \mathcal{B}\right) \mid \mathcal{B}\in L_{2}\right\} $.

\subsection{Contextual, Insertion and Deletion Array Rules}

A $d$\textit{-di\-men\-sio\-nal contextual array rule} (see \cite%
{Freundetal2007}) over the alphabet $V$ is a pair of finite $d$%
-di\-men\-sio\-nal arrays $\left( \left( W_{1},\mathcal{A}_{1}\right)
,\left( W_{2},\mathcal{A}_{2}\right) \right) $ where $W_{1}\cap
W_{2}=\emptyset $ and $shape\left( \mathcal{A}_{1}\right) \cup shape\left( 
\mathcal{A}_{2}\right) \neq \emptyset $. The effect of this contextual rule
is the same as of the array rewriting rule $\left( W_{1}\cup W_{2},\mathcal{A%
}_{1},\mathcal{A}_{1}\cup \mathcal{A}_{2}\right) $, i.e., in the context of $%
\mathcal{A}_{1}$ we insert $\mathcal{A}_{2}$. Hence, such an array rule $%
\left( \left( W_{1},\mathcal{A}_{1}\right) ,\left( W_{2},\mathcal{A}%
_{2}\right) \right) $ can also be called an \textit{array insertion rule},
and then we write $I\left( \left( W_{1},\mathcal{A}_{1}\right) ,\left( W_{2},%
\mathcal{A}_{2}\right) \right) $; if $shape\left( \mathcal{A}_{i}\right)
=W_{i}$, $i\in \left\{ 1,2\right\} $, we simply write $I\left( \mathcal{A}%
_{1},\mathcal{A}_{2}\right) $. Yet we may also interpret the pair $\left(
\left( W_{1},\mathcal{A}_{1}\right) ,\left( W_{2},\mathcal{A}_{2}\right)
\right) $ as having the effect of the array rewriting rule $\mathcal{A}%
_{1}\cup \mathcal{A}_{2}\rightarrow \mathcal{A}_{1}$, i.e., in the context
of $\mathcal{A}_{1}$ we delete $\mathcal{A}_{2}$; in this case, we speak of
an \textit{array deletion rule }and write $D\left( \left( W_{1},\mathcal{A}%
_{1}\right) ,\left( W_{2},\mathcal{A}_{2}\right) \right) $ or, if $%
shape\left( \mathcal{A}_{i}\right) =W_{i}$, $i\in \left\{ 1,2\right\} $,
then even only $D\left( \mathcal{A}_{1},\mathcal{A}_{2}\right) $. For any
(contextual, insertion, deletion) array rule $\left( \left( W_{1},\mathcal{A}%
_{1}\right) ,\left( W_{2},\mathcal{A}_{2}\right) \right) $ we define its
norm by 
\begin{equation*}
\left\Vert \left( \left( W_{1},\mathcal{A}_{1}\right) ,\left( W_{2},\mathcal{%
A}_{2}\right) \right) \right\Vert =\left\Vert W_{1}\cup W_{2}\right\Vert .
\end{equation*}%
The norm of the set of contextual array productions in $G$, $\left\Vert
P\right\Vert $, is defined by 
\begin{equation*}
\left\Vert P\right\Vert =\max \left\{ \,\left\Vert W_{1}\cup
W_{2}\right\Vert \mid \left( \left( W_{1},\mathcal{A}_{1}\right) ,\left(
W_{2},\mathcal{A}_{2}\right) \right) \in P\,\right\} .
\end{equation*}

Let $G_{A}$ be a $d$-\-di\-men\-sio\-nal array grammar $\left( \left[
V^{\ast d}\right] ,\left[ T^{\ast d}\right] ,\left[ \mathcal{A}_{0}\right]
,P,\Longrightarrow _{G_{A}}\right) $ with $P$ containing array insertion and
deletion rules. The norm of the set of array insertion and deletion
rules in $G$, $\left\Vert P\right\Vert $, is defined as for a set of 
contextual array productions, i.e., we again define 
\begin{equation*}
\left\Vert P\right\Vert =\max \left\{ \,\left\Vert W_{1}\cup
W_{2}\right\Vert \mid \left( \left( W_{1},\mathcal{A}_{1}\right) ,\left(
W_{2},\mathcal{A}_{2}\right) \right) \in P\,\right\} .
\end{equation*}%
For $G_{A}$ we consider the array languages $L_{\ast }\left( G_{A}\right) $
and $L_{t}\left( G_{A}\right) $ generated by $G_{A}$ in the modes $\ast $
and $t$, respectively; the corresponding families of array languages are
denoted by $\mathcal{L}_{\mathcal{\delta }}\left( d\text{-}DIA\right) $, $%
\delta \in \left\{ \ast ,t\right\} $; if only array insertion (i.e.,
contextual) rules are used, we have the case of pure grammars, and we also
write $\mathcal{L}_{\mathcal{\delta }}\left( d\text{-}CA\right) $. For
interesting relations between the families of array languages $\mathcal{L}_{%
\mathcal{\ast }}\left( d\text{-}CA\right) $ and $\mathcal{L}_{t}\left( d%
\text{-}CA\right) $ as well as $\mathcal{L}_{\mathcal{\ast }}\left( d\text{-}%
\#\text{-}CFA\right) $ and $\mathcal{L}_{\mathcal{\ast }}\left( d\text{-}%
CFA\right) $ we refer the reader to \cite{Freundetal2007}.

In the following, instead of using the notation $\left( \left[ V^{\ast d}%
\right] ,\left[ T^{\ast d}\right] ,\left[ \mathcal{A}_{0}\right]
,P,\Longrightarrow _{G_{A}}\right) $ for a $d$-dimensional array grammar of
a specific type, we may also simply write $\left( V,T,\left[ \mathcal{A}_{0}%
\right] ,P\right) $. For contextual array grammars or for array grammars
only containing array insertion rules we may even write $\left( V,\left[ 
\mathcal{A}_{0}\right] ,P\right) $.

Our first example shows how we can generate one-dimensional arrays of the
form $LE^{n}\bar{S}E^{m}R$, $n,m\geq 1$, with a contextual array grammar
containing only rules of norm $1$:

\begin{example}
\label{line}Consider the contextual array grammar 
\begin{equation*}
G_{line}=\left( \left\{ \bar{S},E,L,R\right\} ,E\bar{S}E,P\right) \text{
with }P=\left\{ \fbox{$E$}E,E\fbox{$E$},\fbox{$E$}R,L\fbox{$E$}\right\} ;
\end{equation*}%
in order to represent the contextual array (or array insertion and deletion)
rules in a depictive way, the symbols of the selector are enclosed in
boxes). Obviously, $\left\Vert P\right\Vert =1$. Starting from the axiom $E%
\bar{S}E$, the sequence of symbols $E$ is prolonged to the right by the
contextual array rule $\fbox{$E$}E$ and prolonged to the left by the
contextual array rule $E\fbox{$E$}$. The derivation only halts as soon as we
have used both the rules $\fbox{$E$}R$ and $L\fbox{$E$}$ to introduce the
right and left endmarkers $R$ and $L$, respectively. In sum, we obtain 
\begin{equation*}
\left[ L_{t}\left( G_{line}\right) \right] =\left\{ LE^{n}\bar{S}E^{m}R\mid
n,m\geq 1\right\} ,
\end{equation*}%
whereas%
\begin{equation*}
\left[ L_{\ast }\left( G_{line}\right) \right] =\left\{ LE^{n}\bar{S}%
E^{m}R,E^{n}\bar{S}E^{m}R,LE^{n}\bar{S}E^{m},E^{n}\bar{S}E^{m}\mid n,m\geq
1\right\} .
\end{equation*}
\end{example}

\section{(Sequential) P Systems}

For controlling the derivations in an array grammar with array insertion and
deletion rules, in \cite{Freundetal2013} the model of sequential P systems
using array insertion and deletion rules was considered. In general, for
arbitrary types of underlying grammars, P systems are defined as follows:

\smallskip

A \textit{(sequential) P system of type }$X$\textit{\ with tree height }$n$
is a construct $\Pi =\left( G,\mu ,R,i_{0}\right) $ where

\begin{itemize}
\item $G=\left( O,O_{T},A,P,\Longrightarrow _{G}\right) $ is a sequential
grammar of type $X$;

\item $\mu $ is the membrane (tree) structure of the system with the height
of the tree being $n$ ($\mu $ usually is represented by a string containing
correctly nested marked parentheses); we assume the membranes to be the
nodes of the tree representing $\mu $ and to be uniquely labelled by labels
from a set $Lab$;

\item $R$ is a set of rules of the form $\left( h,r,tar\right) $ where $h\in
Lab$, $r\in P$, and $tar$, called the \textit{target indicator}, is taken
from the set $\left\{ here,in,out\right\} \cup \left\{ in_{h}\mid h\in
Lab\right\} $; $R$ can also be represented by the vector $\left(
R_{h}\right) _{h\in Lab}$, where $R_{h}=\left\{ \left( r,tar\right) \mid
\left( h,r,tar\right) \in R\right\} $ is the set of rules assigned to
membrane $h$;

\item $i_{0}$ is the initial membrane containing the axiom $A$.
\end{itemize}

As we only have to follow the trace of a single object during a computation
of the P system, a configuration of $\Pi $ can be described by a pair $%
\left( w,h\right) $ where $w$ is the current object (e.g., string or array)
and $h$ is the label of the membrane currently containing the object $w$.
For two configurations $\left( w_{1},h_{1}\right) $ and $\left(
w_{2},h_{2}\right) $ of $\Pi $ we write $\left( w_{1},h_{1}\right)
\Longrightarrow _{\Pi }\left( w_{2},h_{2}\right) $ if we can pass from $%
\left( w_{1},h_{1}\right) $ to $\left( w_{2},h_{2}\right) $ by applying a
rule $\left( h_{1},r,tar\right) \in R$, i.e., $w_{1}\Longrightarrow
_{r}w_{2} $ and $w_{2}$ is sent from membrane $h_{1}$ to membrane $h_{2}$
according to the target indicator $tar$. More specifically, if $tar=here$,
then $h_{2}=h_{1}$; if $tar=out$, then the object $w_{2}$ is sent to the
region $h_{2}$ immediately outside membrane $h_{1}$; if $tar=in_{h_{2}}$,
then the object is moved from region $h_{1}$ to the region $h_{2}$
immediately inside region $h_{1}$; if $tar=in$, then the object $w_{2}$ is
sent to one of the regions immediately inside region $h_{1}$.

A sequence of transitions between configurations of $\Pi $, starting from
the initial configuration $\left( A,i_{0}\right) $, is called a \textit{%
computation} of $\Pi $. A\textit{\ halting computation} is a computation
ending with a configuration $\left( w,h\right) $ such that no rule from $%
R_{h}$ can be applied to $w$ anymore; $w$ is called the \textit{result} of
this halting computation if $w\in O_{T}$. As the language generated by $\Pi $
we consider $L_{t}\left( \Pi \right) $ which consists of all terminal
objects from $O_{T}$ being results of a halting computation in $\Pi $.

By $\mathcal{L}_{t}\left( X\text{-}LP\right) $ ($\mathcal{L}_{t}\left( X%
\text{-}LP^{\left\langle n\right\rangle }\right) $) we denote the family of
languages generated by P systems (of tree height at most $n$) using grammars
of type $X$. If only the targets $here$, $in$, and $out$ are used, then the
P system is called \textit{simple}, and the corresponding families of
languages are denoted by $\mathcal{L}_{t}\left( X\text{-}LsP\right) $ ($%
\mathcal{L}_{t}\left( X\text{-}LsP^{\left\langle n\right\rangle }\right) $).

In the string case (see \cite{Freundetal2012}), the operations of left and
right insertion ($I)$ of strings of length $m$ and left and right deletion ($%
D$) of strings of length $k$ were investigated; the corresponding types are
abbreviated by $D^{k}I^{m}$. Every language $L\subseteq T^{\ast }$ in $%
\mathcal{L}_{\ast }\left( D^{1}I^{1}\right) $ can be written in the form $%
T_{l}^{\ast }ST_{r}^{\ast }$ where $T_{l},T_{r}\subseteq T$ and $S$ is a
finite subset of $T^{\ast }$. Using the regulating mechanism of P systems,
we get $\left\{ a^{2^{n}}\mid n\geq 0\right\} \in \mathcal{L}_{t}\left(
D^{1}I^{2}\text{-}LP^{\left\langle 1\right\rangle }\right) $ and even obtain
computational completeness:

\begin{theorem}
\label{theoremstrings}(see \cite{Freundetal2012}) $\mathcal{L}_{t}\left(
D^{1}I^{1}\text{-}LsP^{\left\langle 8\right\rangle }\right) =RE$.
\end{theorem}

One-dimensional arrays can also be interpreted as strings; left/right
insertion of a symbol $a$ corresponds to taking the set containing all rules 
$I\left( 
\begin{array}[b]{cc}
a & \fbox{${b}$}%
\end{array}%
\right) $/$I\left( 
\begin{array}[b]{cc}
\fbox{${b}$} & a%
\end{array}%
\right) $ for any $b$; left/right deletion of a symbol $a$ corresponds to
taking the rule $D\left( 
\begin{array}[b]{cc}
\fbox{${\#}$} & a%
\end{array}%
\right) $/$D\left( 
\begin{array}[b]{cc}
a & \fbox{${\#}$}%
\end{array}%
\right) $; these array insertion and deletion rules have norm one, but the
array deletion rules also sense for the blank symbol $\#$ in the selector.
Hence, from Theorem~\ref{theoremstrings}, we immediately infer the following
result, with $D^{k}I^{m}A$ denoting the type of array grammars using array
deletion and insertion rules of norms at most $k$ and $m$, respectively:

\begin{corollary}
\label{theorem1arrays}$\mathcal{L}_{t}\left( 1\text{-}D^{1}I^{1}A\text{-}%
LsP^{\left\langle 8\right\rangle }\right) =\mathcal{L}_{\ast }\left( 1\text{-%
}ARBA\right) $.
\end{corollary}

\medskip

With respect to the tree height of the simple P systems, this result will be
improved considerably in Section~\ref{comp}. One-\-di\-men\-sio\-nal array
grammars with only using array insertion and deletion rules of norm at most
two will be shown to be computationally complete in Section~\ref{comp2}.

\subsection{Encoding the Post Correspondence Problem With Array Insertion P
Systems}

An instance of the \textit{Post Correspondence Problem} is a pair of
sequences of non-empty strings $\left( u_{1},\ldots ,u_{n}\right) $ and $%
\left( v_{1},\ldots ,v_{n}\right) $ over an alphabet $T$. A solution of this
instance is a sequence of indices $i_{1},\ldots ,i_{k}$ such that $%
u_{i_{1}}\ldots u_{i_{k}}=v_{i_{1}}\ldots v_{i_{k}}$; we call $%
u_{i_{1}}\ldots u_{i_{k}}$ the result of this solution. Let 
\begin{equation*}
\left\{ u_{i_{1}}\ldots u_{i_{k}}\mid i_{1},\ldots ,i_{k}\text{ is a
solution of }\left( \left( u_{1},\ldots ,u_{n}\right) ,\left( v_{1},\ldots
,v_{n}\right) \right) \right\}
\end{equation*}%
be the set of results of all solutions of the instance $\left( \left(
u_{1},\ldots ,u_{n}\right) ,\left( v_{1},\ldots ,v_{n}\right) \right) $ of
the Post Correspondence Problem, denoted by $L\left( \left( u_{1},\ldots
,u_{n}\right) ,\left( v_{1},\ldots ,v_{n}\right) \right) $.

We now show how $L\left( \left( u_{1},\ldots ,u_{n}\right) ,\left(
v_{1},\ldots ,v_{n}\right) \right) $ can be represented in a very specific
way as the language generated by an array insertion P system. Consider the
homomorphism $h_{\Sigma }$ defined by $h_{\Sigma }:\Sigma \rightarrow 
\Sigma \Sigma ^{\prime }$ with $h_{\Sigma }\left( a\right) =aa^{\prime }$ 
for all $a\in \Sigma $.

\begin{lemma}
Let $I=\left( \left( u_{1},\ldots ,u_{n}\right) ,\left( v_{1},\ldots
,v_{n}\right) \right) $ be an instance of the Post Correspondence Problem
over $T$. Then we can effectively construct a one-dimensional array
insertion P system $\Pi $ such that 
\begin{equation*}
\left[ L\left( \Pi \right) \right] =\left\{ LL^{\prime }h_{T}\left( w\right)
RR^{\prime }\mid w\in L\left( \left( u_{1},\ldots ,u_{n}\right) ,\left(
v_{1},\ldots ,v_{n}\right) \right) \right\} .
\end{equation*}
\end{lemma}

\begin{proof}
The main idea for constructing the one-dimensional array insertion P system%
\begin{equation*}
\Pi =\left( G,\left[ \hspace*{0.01cm}_{_{0}}\left[ \hspace*{0.01cm}_{_{1}}%
\hspace*{0.01cm}\hspace*{0.1cm}\right] \hspace*{0.01cm}_{_{1}}\ldots \left[ 
\hspace*{0.01cm}_{_{n}}\hspace*{0.01cm}\hspace*{0.1cm}\right] \hspace*{0.01cm%
}_{_{n}}\ldots \left[ \hspace*{0.01cm}_{_{n+1}}\hspace*{0.01cm}\hspace*{0.1cm%
}\right] \hspace*{0.01cm}_{_{n+1}}\right] \hspace*{0.01cm}_{_{0}},R,0\right)
\end{equation*}%
with 
\begin{equation*}
G=\left( \left\{F, L,R,L^{\prime },R^{\prime }\right\} \cup T\cup T^{\prime
},LL^{\prime },P\cup \left\{ \fbox{$a$}\fbox{$a^{\prime }$}RR^{\prime }\mid
a\in T\right\} \right)
\end{equation*}%
is to generate sequences $u_{i_{1}}\ldots u_{i_{k}}$ and $\left( v_{i_{1}}
\right) ^{\prime } \ldots \left( v_{i_{k}}\right) ^{\prime }$ for sequences of indices
$i_{1},\ldots ,i_{k}$ in an interleaving way, i.e., each symbol $a\in T$ from
the first sequence is followed by the corresponding primed symbol $a^{\prime
}$ in the second sequence; as soon as both sequences have reached the same
length, i.e., if we have got the encoding of a solution for this instance of
the Post Correspondence Problem, we may use a fitting contextual array rule $%
\fbox{$a$}$ $\fbox{$a^{\prime }$}RR^{\prime }$ for some $a\in T$ to stop the
derivation in $\Pi $.

For generating these interleaving sequences of symbols $a$ and $a^{\prime }$
we take the following rules into $P$:

For prolonging the first sequence by $u_i$, 
$u_i=u_{i}\left( 1\right) \ldots u_{i}\left( \left\vert u_{i}\right\vert \right) $,
we add all rules of the form 
\begin{quotation}
$\fbox{$u_{i}\left( 0\right) $}\fbox{$u_{i}\left( 0\right) ^{\prime \prime }$}%
u_{i}\left( 1\right) \fbox{$u_{i}\left( 1\right) ^{\prime \prime }$}\ldots
u_{i}\left( \left\vert u_{i}\right\vert \right) \fbox{$u_{i}\left(
\left\vert u_{i}\right\vert \right) ^{\prime \prime }$}$
\end{quotation}%
with $u_{i}\left( 0\right) \in T\cup \left\{ L\right\} $ and $u_{i}\left(
j\right) ^{\prime \prime }\in \left\{ u_{i}\left( j\right) ^{\prime
},\#\right\} $, $0\leq j\leq \left\vert u_{i}\right\vert $, $1\leq i\leq n$,
fulfilling the following constraints:

-- if $u_{i}\left( 0\right) =L$ then $u_{i}\left( 0\right) ^{\prime \prime
}=L^{\prime }$;

-- if $u_{i}\left( k\right) ^{\prime \prime }=\#$ for some $k\geq 0$, then $%
u_{i}\left( j\right) ^{\prime \prime }=\#$ for all $k\leq j\leq \left\vert
u_{i}\right\vert $;

-- if $u_{i}\left( k\right) ^{\prime \prime }=u_{i}\left( k\right) ^{\prime
} $ for some $k\geq 0$, then $u_{i}\left( j\right) ^{\prime \prime
}=u_{i}\left( j\right) ^{\prime }$ for all $j\leq k$.

In the same way, for prolonging the second sequence by $v_i$, 
$v_i=v_{i}\left( 1\right) \ldots v_{i}\left( \left\vert v_{i}\right\vert \right) $,
represented with primed symbols, we add all rules of the form 
\begin{quotation}
$\fbox{$v_{i}\left( 0\right) ^{\prime \prime }$}\fbox{$v_{i}\left( 0\right)
^{\prime }$}\fbox{$v_{i}\left( 1\right) ^{\prime \prime }$}v_{i}\left(
1\right) ^{\prime }\ldots \fbox{$v_{i}\left( \left\vert v_{i}\right\vert
\right) ^{\prime \prime }$}v_{i}\left( \left\vert v_{i}\right\vert \right)
^{\prime } $
\end{quotation}%
with $v_{i}\left( 0\right) \in T\cup \left\{ L\right\} $ and $v_{i}\left(
j\right) ^{\prime \prime }\in \left\{ v_{i}\left( j\right) ,\#\right\} $, $%
0\leq j\leq \left\vert v_{i}\right\vert $, $1\leq i\leq n$, fulfilling the
following constraints:

-- if $v_{i}\left( 0\right) ^{\prime }=L^{\prime }$ then $v_{i}\left(
0\right) ^{\prime \prime }=L$;

-- if $v_{i}\left( k\right) ^{\prime \prime }=\#$ for some $k\geq 0$, then $%
v_{i}\left( j\right) ^{\prime \prime }=\#$ for all $k\leq j\leq \left\vert
v_{i}\right\vert $;

-- if $v_{i}\left( k\right) ^{\prime \prime }=v_{i}\left( k\right) $ for
some $k\geq 0$, then $v_{i}\left( j\right) ^{\prime \prime }=v_{i}\left(
j\right) $ for all $j\leq k$.

The set $R$ consists of the following rules:

Starting from the axion $LL^{\prime }$, in membrane region $0$ we have all
rules 
\begin{quotation}
$\left( 0,I\left( \fbox{$u_{i}\left( 0\right) $}\fbox{$u_{i}\left( 0\right)
^{\prime \prime }$}u_{i}\left( 1\right) \fbox{$u_{i}\left( 1\right) ^{\prime
\prime }$}\ldots u_{i}\left( \left\vert u_{i}\right\vert \right) \fbox{$%
u_{i}\left( \left\vert u_{i}\right\vert \right) ^{\prime \prime }$}\right)
,in_{i}\right) ,$
\end{quotation}
i.e., when adding the sequence corresponding to the string $u_{i}$, the
resulting array is sent into membrane $i$, where the sequence of primed
strings corresponding to the string $v_{i}$ is added and the resulting array
is sent out again into the skin region $0$ using any of the rules%
\begin{quotation}
$\left( i,I\left( \fbox{$v_{i}\left( 0\right) ^{\prime \prime }$}\fbox{$%
v_{i}\left( 0\right) ^{\prime }$}\fbox{$v_{i}\left( 1\right) ^{\prime \prime
}$}v_{i}\left( 1\right) ^{\prime }\ldots \fbox{$v_{i}\left( \left\vert
v_{i}\right\vert \right) ^{\prime \prime }$}v_{i}\left( \left\vert
v_{i}\right\vert \right) ^{\prime }\right) ,out\right) .$
\end{quotation}

For the cases when no fitting rules for prolonging the array exist, we take
the rules 

\begin{quotation}
$\left( 0,I\left( \fbox{$X$}\fbox{$\#$} F\right) ,here\right) $
\end{quotation}
for any $X\in T\cup \left\{ F\right\}$ and 

\begin{quotation}
$\left( i,I\left( \fbox{$X$}\fbox{$\#$} F\right) ,here\right) $
\end{quotation}
for any $X\in T^{\prime }\cup \left\{ L^{\prime } ,F\right\}$ and $1\leq
i\leq n$.

The observation that (only) the application of an array insertion rule 
\begin{quotation}
$\left( 0,I\left( \fbox{$a$}\fbox{$a^{\prime }$}RR^{\prime }\right)
,in_{n+1}\right) $ 
\end{quotation}
for some $a\in T$ stops the derivation, with sending the
terminal array into membrane $n+1$, completes the proof.
$\hfill {}$
\end{proof}

\bigskip

As is well known (see \cite{Post46}), the Post Correspondence Problem is
undecidable, hence, the emptiness problem for $\mathcal{L}_{t}\left( DI\text{%
-}LP^{\left\langle 1\right\rangle }\right) $ is undecidable:

\begin{corollary}
For any $k\geq 1$, the emptiness problem for $\mathcal{L}_{t}\left( DI\text{-%
}LP^{\left\langle k\right\rangle }\right) $ is undecidable.
\end{corollary}

\smallskip

For $d\geq 2$, even the emptiness problem for $\mathcal{L}_{t}\left( d\text{-%
}CA\right) $ is undecidable, which follows from the result obtained in \cite%
{Fernauetal2005}, where it was shown that every recursively enumerable
one-dimensional array language can be characterized as the projection of an
array language generated by a two-dimensional contextual array grammar using
rules of norm one only.

\section{Computational Completeness of Array Insertion and Deletion P
Systems Using Rules with Norm at Most One\label{comp}}

We now show the first of our main results: any recursively enumerable
one-di\-men\-sio\-nal array language can be generated by an array insertion
and deletion P system which only uses rules of norm at most one and the
targets $here$, $in$, and $out$ and whose membrane structure has only tree
height $2$; for two-dimensional array languages, the corresponding result
was established in \cite{Freundetal2013}.

\begin{theorem}
\label{main}$\mathcal{L}_{t}\left( 1\text{-}D^{1}I^{1}A\text{-}%
LsP^{\left\langle 2\right\rangle }\right) =\mathcal{L}_{\ast }\left( 1\text{-%
}ARBA\right) $.
\end{theorem}

\begin{proof}
The main idea of the proof is to construct the simple P system $\Pi $ of
type $1$-$DIA$ with a membrane structure of height two generating a
recursively enumerable one-di\-men\-sio\-nal array language $L_{A}$ given
by a grammar $G_{A}$ of type $1$-$ARBA$ in such a way that we first generate
the \textquotedblleft workspace\textquotedblright , i.e., the lines as
described in Example~\ref{line} and then simulate the rules of the
one-di\-men\-sio\-nal array grammar $G_{A}$ inside this \textquotedblleft
workspace\textquotedblright ; finally, the superfluous symbols $E$ and $L,R$
have to be erased to obtain the terminal array.

Now let $G_{A}=\left( \left[ \left( N\cup T\right) ^{\ast 1}\right] ,\left[
T^{\ast 1}\right] ,\left[ \mathcal{A}_{0}\right] ,P,\Longrightarrow
_{G_{A}}\right) $ be an array grammar of type $1$-$ARBA$ generating $L_{A}$.
In order to make the simulation in $\Pi $ easier, without loss of
generality, we may make some assumptions about the forms of the array rules 
in $P$: First of all, we may assume that the array rules are in a kind of
Chomsky normal form (e.g., compare \cite{Freund93}), i.e., only of the
following forms: $A\rightarrow B$ for $A\in N\ $and $B\in N\cup T\cup
\left\{ \#\right\} $ as well as $AvD\rightarrow BvC$ with $\left\Vert
v\right\Vert =1$ (i.e., $v\in \left\{ \left( 1\right) ,\left( -1\right)
\right\} $), $A,B,C\in N\cup T$, and $D\in N\cup T\cup \left\{ \#\right\} $
(we would like to emphasize that usually $A,B,C,D$ in the array rule 
$AvD\rightarrow BvC$ are not allowed to be terminal symbols); in a
more formal way, the rule $AvD\rightarrow BvC$ represents the rule $\left( W,%
\mathcal{A}_{1},\mathcal{A}_{2}\right) $ with $W=\left\{ \left( 0\right)
,v\right\} $, $\mathcal{A}_{1}=\left\{ \left( \left( 0\right) ,A\right)
,\left( v,D\right) \right\} $, and $\mathcal{A}_{2}=\left\{ \left( \left(
0\right) ,B\right) ,\left( v,C\right) \right\} $. As these rules in fact are
simulated in $\Pi $ with the symbol $E$ representing the blank symbol $\#$,
a rule $Av\#\rightarrow BvC$ now corresponds to a rule $AvE\rightarrow BvC$.
Moreover, a rule $A\rightarrow B$ for $A\in N\ $and $B\in N\cup T$ can be
replaced by the set of all rules $AvD\rightarrow BvD$ for all $D\in N\cup
T\cup \left\{ E\right\} $ and $v\in \left\{ \left( 1\right) ,\left(
-1\right) \right\} $, and $A\rightarrow \#$ can be replaced by the set of
all rules $AvD\rightarrow EvD$ for all $D\in N\cup T\cup \left\{ E\right\} $
and $v\in \left\{ \left( 1\right) ,\left( -1\right) \right\} $.

After these replacements described above, in the P system $\Pi $ we now only
have to simulate rules of the form $AvD\rightarrow BvC$ with $v\in \left\{
\left( 1\right) ,\left( -1\right) \right\} $ as well as $A,B,C,D\in N\cup
T\cup \left\{ E\right\} $. Yet in order to obtain a P system $\Pi $ with the
required features, we make another assumption for the rules to be simulated:
any intermediate array obtained during a derivation contains exactly one
symbol marked with a bar; as we only have to deal with sequential systems
where at each moment exactly one rule is going to be applied, this does not
restrict the generative power of the system as long as we can guarantee that
the marking can be moved to any place within the current array. Instead of a
rule $AvD\rightarrow BvC$ we therefore take the corresponding rule $\bar{A}%
vD\rightarrow Bv\bar{C}$; moreover, to move the bar from one position in the
current array to another position, we add all rules $\bar{A}vC\rightarrow Av%
\bar{C}$ for all $A,C\in N\cup T\cup \left\{ E\right\} $ and $v\in \left\{
\left( 1\right) ,\left( -1\right) \right\} $. We collect all these rules
obtained in the way described so far in a set of array rules $P^{\prime }$
and assume them to be uniquely labelled by labels from a set of labels $%
Lab^{\prime }$, i.e., $P^{\prime }=\left\{ l:\bar{A}_{l}vD_{l}\rightarrow
B_{l}v\bar{C}_{l}\mid l\in Lab^{\prime }\right\} $.

After all these preparatory steps we now are able to construct the simple P
system $\Pi $ with array insertion and deletion rules:
\begin{equation*}
\Pi =\left( G,\left[ \hspace*{0.01cm}_{_{0}}\left[ \hspace*{0.01cm}%
_{_{I_{1}}}\left[ \hspace*{0.01cm}_{_{I_{2}}}\hspace*{0.1cm}\right] \hspace*{%
0.01cm}_{_{I_{2}}}\hspace*{0.1cm}\right] \hspace*{0.01cm}_{_{I_{1}}}\ldots %
\left[ \hspace*{0.01cm}_{_{l_{1}}}\left[ \hspace*{0.01cm}_{_{l_{2}}}\hspace*{%
0.1cm}\right] \hspace*{0.01cm}_{_{l_{2}}}\hspace*{0.1cm}\right] \hspace*{%
0.01cm}_{_{l_{1}}}\ldots \left[ \hspace*{0.01cm}_{_{F_{1}}}\left[ \hspace*{%
0.01cm}_{_{F_{2}}}\hspace*{0.1cm}\right] \hspace*{0.01cm}_{_{F_{2}}}\hspace*{%
0.1cm}\right] \hspace*{0.01cm}_{_{F_{1}}}\right] \hspace*{0.01cm}%
_{_{0}},R,I_{2}\right)
\end{equation*}
with $I_{1}$ and $I_{2}$ being the membranes for generating the initial
lines, $F_{1}$ and $F_{2}$ are the membranes to extract the final terminal
arrays in halting computations, and $l_{1}$ and $l_{2}$ for all $l\in
Lab^{\prime }$ are the membranes to simulate the corresponding array rule
from $P^{\prime }$ labelled by $l$. The components of the underlying array
grammar $G$ can easily be collected from the description of the rules in $R$
as described below.

We start with the initial array $\mathcal{A}_{0}=E\bar{S}E$ from 
Example~\ref{line} and take all rules 

\begin{quotation}
$\left( I_{2},I\left( r\right) ,here\right) $ 
\end{quotation}

\noindent with all rules $r\in \left\{ \fbox{$E$}E,E\fbox{$E$}\right\} $, 
taken as array insertion rules; using the insertion rule 

\begin{quotation}
$\left( I_{2},I\left( L\fbox{$E$}\right) ,out\right) $ 
\end{quotation}
we get out of membrane $I_{2}$ into membrane region $%
I_{1}$, and by using 

\begin{quotation}
$\left( I_{1},I\left( \fbox{$E$}R\right) ,out\right) $
\end{quotation}
we move the initial line $LE^{n}\bar{S}E^{m}R$ for some $n,m\geq 1$ out into
the skin membrane.
\medskip

To be able to simulate a derivation from $G_{A}$ for a specific terminal
array, the workspace in this initial line has to be large enough, but as we
can generate such lines with arbitrary size, such an initial array can be
generated for any terminal array in $L_{\ast }\left( G_{A}\right) $.
\medskip

An array rule from $P^{\prime }=\left\{ l:\bar{A}_{l}vD_{l}\rightarrow B_{l}v%
\bar{C}_{l}\mid l\in Lab^{\prime }\right\} $ is simulated by applying the
following sequence of array insertion and deletion rules in the membranes $%
l_{1}$ and $l_{2}$, which send the array twice the path from the skin
membrane to membrane $l_{2}$ via membrane $l_{1}$ and back to the skin
membrane:
\medskip

\begin{quotation}
$\left( 0,I\left( \fbox{${R}$}K_{l}\right) ,in\right) $, 
$\left( l_{1},D\left( \fbox{${\bar{A}}_{l}$}vD_{l}\right) ,in\right) $, 
\end{quotation}

\begin{quotation}
$\left( l_{2},I\left( \fbox{${\bar{A}}_{l}$}v\bar{D}_{l}^{\left( l\right) }\right)
,out\right) $, $\left( l_{1},D\left( \fbox{$\bar{D}_{l}^{\left( l\right) }$}%
\left( -v\right) {\bar{A}}_{l}\right) ,out\right) $,
\end{quotation}

\begin{quotation}
$\left( 0,I\left( \fbox{$\bar{D}_{l}^{\left( l\right) }$}\left( -v\right)
B_{l}\right) ,in\right) $, $\left( l_{1},D\left( \fbox{$B_{l}$}v\bar{D}%
_{l}^{\left( l\right) }\right) ,in\right) $,
\end{quotation}

 \begin{quotation}
$\left( l_{2},D\left( \fbox{${R}
$}K_{l}\right) ,out\right) $, $\left( l_{1},I\left( \fbox{$B_{l}$}v{\bar{C}}%
_{l}\right) ,out\right) $.
\end{quotation}

\medskip

Whenever reaching the skin membrane, the current array contains exactly one
barred symbol. If we reach any of the membranes $l_{1}$ and/or $l_{2}$ with
the wrong symbols (which implies that none of the rules listed above is
applicable), we introduce the trap symbol $F$ by the rules 

\begin{quotation}
$\left( m,I\left(
F\fbox{${L}$}\right) ,out\right) $ and $\left( m,I\left( F\fbox{${F}$}%
\right) ,out\right) $ 
\end{quotation}

\noindent for $m\in \left\{ l_{1},l_{2}\mid l\in Lab^{\prime
}\cup \left\{ I\right\} \right\} $; as soon as $F$ has been introduced once,
with 

\begin{quotation}
$\left( 0,I\left( F\fbox{${F}$}\right) ,in\right) $ 
\end{quotation}

\noindent we can guarantee that the computation in $\Pi $ will never stop.

As soon as we have obtained an array representing a terminal array, the
corresponding array computed in $\Pi $ is moved into membrane $F_{1}$ by the
rule $\left( 0,D\left( R\right) ,in\right) $ (for any $X$, $D\left( X\right) 
$ / $I\left( K\right) $ just means deleting/inserting $X$ without taking
care of the context). In membrane $F_{1}$, the left endmarker $L$ and all
superfluous symbols $E$ as well as the marked blank symbol $\bar{E}$
(without loss of generality we may assume that at the end of the simulation
of a derivation from $G_{A}$ in $\Pi $ the marked symbol is $\bar{E}$) are
erased by using the rules $\left( F_{1},D\left( X\right) ,here\right) $ with 
$X\in \left\{ E,\bar{E},L\right\} $. The computation in $\Pi $ halts with
yielding a terminal array in membrane $F_{1}$ if and only if no other
non-terminal symbols have occurred in the array we have moved into $F_{1}$;
in the case that non-terminal symbols occur, we start an infinite loop
between membrane $F_{1}$ and membrane $F_{2}$ by introducing the trap symbol $F$: 
\begin{quotation}
$\left( F_{1},D\left( X\right) ,in\right) $ for $X\notin T\cup \left\{
E,\bar{E},L\right\} $ and $\left( F_{2},I\left( F\right) ,out\right) $.
\end{quotation}

As can be seen from the description of the rules in $\Pi $, we can simulate
all terminal derivations in $G_{A}$ by suitable computations in $\Pi $, and
a terminal array $\mathcal{A}$ is obtained as the result of a halting
computation (always in membrane $F_{1}$) if and only if $\mathcal{A}\in
L_{\ast }\left( G_{A}\right) $; hence, we conclude $L_{t}\left( \Pi \right)
=L_{\ast }\left( G_{A}\right) $.
$\hfill {}$
\end{proof}

\section{Computational Completeness of Array Grammars with Array Insertion
and Deletion Rules with Norms of at Most Two\label{comp2}}

When allowing array insertion and deletion rules with norms of at most two,
computational completeness can even be obtained without any additional
control mechanism (as for example, using P systems as considered in Section~%
\ref{comp}).

\begin{theorem}
\label{main2}$\mathcal{L}_{t}\left( 1\text{-}D^{2}I^{2}A\right) =\mathcal{L}%
_{\ast }\left( 1\text{-}ARBA\right) $.
\end{theorem}

\begin{proof}
The main idea of the proof is to construct a one-dimensional array grammar
with array insertion and deletion rules simulating the actions of a Turing
machine $M_{A}$ with a bi-infinite tape which generates a given
one-dimensional array language (we identify each position of the tape with
the corresponding position in $\mathbb{Z}$; in that sense, $M_{A}$ can be
seen as a machine generating arrays). Let $M_{A}=\left( Q,V,T,\delta
,q_{0},q_{f}\right) $ where $Q$ is a finite set of states, $V$ is the tape
alphabet, $T\subseteq V$ is the input alphabet, $\delta $ is the transition
function, $q_{0}$ is the initial state, and $q_{f}$ is the final state. The
Turing machine starts on the empty tape, which means that on each position
there is the blank symbol represented by the special symbol $E\in V$.

We now construct the one-dimensional array grammar 
$G_{A}=\left( V^{\prime
},T,P,A\right) $ with
\begin{eqnarray*}
V^{\prime } & = & V\cup \left\{ L,R,L^{\prime },R^{\prime },E^{\prime
},F\right\} \\ 
& \cup & \left\{ \left[ AqXD\right] \mid A\in V\cup \left\{ L\right\} ,q\in
Q\cup \left\{ q_{f}^{\prime }\right\} ,X\in V,D\in V\cup \left\{ R\right\}
\right\} ,\\
A&=&LE\left[ Eq_{0}EE\right] ER
\end{eqnarray*}

\noindent and with the set of array insertion and deletion
rules $P$ constructed according to the following \textquotedblleft
program\textquotedblright :

The simulation of a computation of $M_{A}$ starts with the axiom $LE\left[
Eq_{0}EE\right] ER$; $L$ and $R$ are the left and the right endmarker,
respectively. Throughout the whole simulation, the position of the
(read/write-)head of the Turing machine $M_{A}$ is marked by the special
symbol $\left[ AqXD\right] $ indicating that the head currently is on a
symbol $X$ with an $A$ to its left and a $D$ to its right. Whenever new
\textquotedblleft workspace\textquotedblright\ is needed, $L$ or $R$ are
moved one position to the left or right, respectively, at the same time
inserting another $E$:

\begin{quotation}
$I\left( \fbox{$\left[ AqXE\right] $}\fbox{${R}$}R\right) $, $D\left( 
\fbox{$\left[ AqXE\right] $}R\fbox{${R}$}\right) $, $I\left( \fbox{$\left[
AqXE\right] $}E\fbox{${R}$}\right) $;
\end{quotation}

\begin{quotation}
$I\left( L\fbox{${L}$}\fbox{$\left[ EqXD\right] $}\right) $, $D\left( 
\fbox{${L}$}L\fbox{$\left[ EqXD\right] $}\right) $, $I\left( \fbox{${L}$}E%
\fbox{$\left[ EqXD\right] $}\right) $.
\end{quotation}

Any transition $\left( p,Y,R\right) \in \delta \left( q,X\right) $ (reading $%
X$ in state $q$, $M_{A}$ enters state $p$, rewrites $X$ by $Y$ and moves its
head one position to the right) is simulated by the rules

\begin{quotation}
$D\left( \fbox{${A}$}\fbox{$\left[ AqXD\right] $}D\right) $, 
$I\left( \fbox{$\left[ AqXD\right] $}\left[ YpDC\right] \fbox{${C}$}\right) $, 
\end{quotation}
\begin{quotation}
$D\left( \left[ AqXD\right] \fbox{$\left[ YpDC\right] $}\fbox{${C}$}\right) $, 
$I\left( Y\fbox{$\left[ YpDC\right] $}\fbox{${C}$}\right) $;
\end{quotation}

\noindent any transition $\left( p,Y,L\right) \in \delta \left( q,X\right) $ (reading 
$X$ in state $q$, $M_{A}$ enters state $p$, rewrites $X$ by $Y$ and moves 
its head one position to the left) is simulated by the rules

\begin{quotation}
$D\left( A\fbox{$\left[ AqXD\right] $}\fbox{${D}$}\right) $, 
$I\left( \fbox{${C}$}\left[ CpAY\right] \fbox{$\left[ AqXD\right] $}\right) $, 
\end{quotation}
\begin{quotation}
$D\left( \fbox{${C}$}\fbox{$\left[ CpAY\right] $}\left[ AqXD\right] \right) $, 
$I\left( \fbox{${C}$}\fbox{$\left[ CpAY\right] $}Y\right) $.
\end{quotation}

As soon as $M_{A}$ has reached the final state $q_{f}$ (without loss of
generality, we may assume that this state is the only one where $M_{A}$ may
halt), we start the final procedure in $G_{A}$ to obtain the terminal array:
First, we go to the left until $q_{f}$ reaches the left border $L$ of the
workspace, delete the left endmarker $L$ and go into the state $%
q_{f}^{\prime }$:

\begin{quotation}
$I\left( L^{\prime }\fbox{${L}$}\fbox{$\left[ Lq_{f}XD\right] $}\right) $, $%
D\left( \fbox{${L}^{\prime }$}L\fbox{$\left[ Lq_{f}XD\right] $}\right) $,
\end{quotation}
\begin{quotation}
$I\left( \fbox{${L}^{\prime }$}\left[ Eq_{f}^{\prime }EX\right] \fbox{$\left[
Lq_{f}XD\right] $}\right) $, $D\left( \fbox{${L}^{\prime }$}\fbox{$\left[
Eq_{f}^{\prime }EX\right] $}\left[ Lq_{f}XD\right] \right) $,
\end{quotation}
\begin{quotation}
$I\left( \fbox{${L}^{\prime }$}\fbox{$\left[ Eq_{f}^{\prime }EX\right] $}%
X\right) $, $D\left( L^{\prime }\fbox{$\left[ Eq_{f}^{\prime }EX\right] $}%
\fbox{$X$}\right) $, $I\left( E^{\prime }\fbox{$\left[ Eq_{f}^{\prime }EX%
\right] $}\fbox{$X$}\right) $.
\end{quotation}

With state $q_{f}^{\prime }$, $G_{A}$ now goes to the right, keeping each $%
a\in T$, but replacing $E$ by $E^{\prime }$ (these rules are the same as if
simulating the corresponding transitions in a Turing machine, hence, we do
not specify them here) until the right endmarker $R$ is reached:

\begin{quotation}
$I\left( \fbox{$\left[ Aq_{f}^{\prime }ER\right] $}\fbox{${R}$}R^{\prime
}\right) $, $D\left( \fbox{$\left[ Aq_{f}^{\prime }ER\right] $}R\fbox{${R}%
^{\prime }$}\right) $,
\end{quotation}

\begin{quotation}
 $I\left( \fbox{$\left[ Aq_{f}^{\prime }ER\right] $}%
E^{\prime }\fbox{${R}^{\prime }$}\right) $, $D\left( \left[ Aq_{f}^{\prime
}ER\right] \fbox{$E^{\prime }$}\fbox{${R}^{\prime }$}\right) $,
\end{quotation}

\begin{quotation}
$I\left( E^{\prime }\fbox{$E^{\prime }$}\fbox{${R}^{\prime }$}\right) $, 
$D\left( \fbox{$E^{\prime }$}\fbox{${E}^{\prime }$}R^{\prime }\right) $.
\end{quotation}

After the deletion of $R^{\prime }$, only the symbols $E^{\prime }$ remain
to be deleted by the array deletion rule $D\left( E\right) $.

Whenever something goes wrong in the process of simulating the transitions
of $M_{A}$ in $G_{A}$, the application of a trap rule will be enforced,
yielding an unbounded sequence of trap symbols $F$ to the right:

\begin{quotation}
$I\left( \fbox{$R$}F\right) $, $I\left( \fbox{$R^{\prime }$}F\right) $, $%
I\left( \fbox{$F$}F\right) $.
\end{quotation}

As can be seen from the description of the rules in $G_{A}$, we can simulate
all terminal computations in $M_{A}$ by suitable derivations in $G_{A}$, and
a terminal array $\mathcal{A}$ is obtained as the result of a halting
computation if and only if $\mathcal{A}\in L_{t}\left( G_{A}\right) $;
hence, we conclude $L_{t}\left( G_{A}\right) =L\left( M_{A}\right) $.
$\hfill {}$
\end{proof}

\section{Conclusion}

Array insertion grammars have already been considered as contextual array
grammars in \cite{Freundetal2007}, whereas the inverse interpretation of a
contextual array rule as a deletion rule has newly been introduced in \cite%
{Fernauetal2013}, which continued the research on P systems with left and
right insertion and deletion of strings, see \cite{Freundetal2012}.

In the main part of our paper, we have restricted ourselves to exhibit
examples of one-\-di\-men\-sio\-nal array languages that can be generated by
array insertion (contextual array) grammars as well as to show that array
insertion and deletion P systems using rules with norm at most one and even
array grammars only using array insertion and deletion rules with norm at
most two are computationally complete.

In \cite{Fernauetal2013}, the corresponding computational completeness
result has been shown for two-\-di\-men\-sio\-nal array insertion and
deletion P systems using rules with norm at most one. It remains as an
interesting question for future research whether the result for array
grammars only using array insertion and deletion rules with norm at most two
can also be achieved in higher dimensions, but at least for dimension two.

\subsubsection*{Acknowledgements}

We gratefully acknowledge the interesting discussions with Henning Fernau
and Markus L. Schmid on many topics considered in this paper, which is the
continuation of our joint research elaborated in \cite{Fernauetal2013}.

\bibliographystyle{eptcs}
\bibliography{MCU2013ArraysOneFinal}

\end{document}